\documentclass[final,authoryear,3p,times,twocolumn]{elsarticle}

\usepackage{amssymb}
\usepackage{amsmath}
\usepackage{amsfonts}
\usepackage{graphicx}      
\usepackage{natbib}        
\usepackage[usenames]{color}
\usepackage{textcomp}
\usepackage{mathtools}

\newtheorem{theorem}{Theorem}
\newtheorem{prop}[theorem]{Proposition}
\newtheorem{lemma}[theorem]{Lemma}

\newtheorem{definition}[theorem]{Definition}

\newtheorem{remark}[theorem]{Remark}
\newtheorem{example}[theorem]{Example}
\newenvironment{proof}[1][Proof]{\textbf{#1.} }{\ \rule{0.5em}{0.5em}}

\journal{arXiv}

\begin{document}

\begin{frontmatter}



\title{Boundary controlled irreversible port-Hamiltonian systems}


\author[First]{Hector Ramirez} 
\author[Second]{Yann Le Gorrec}
\author[Third]{Bernhard Maschke} 
\address[First]{Universidad T\'{e}cnica Federico Santa Mar\'{i}a, Valparaiso, Chile. (hector.ramireze@usm.cl).}
\address[Second]{D\'epartement d'Automatique et Syst\`emes Micro-M\'ecatroniques, FEMTO-ST UMR CNRS 6174, Universit\'{e} de Bourgogne Franche Comt\'{e}, 26 chemin de l'\'epitaphe, F-25030 Besan\c{c}on, France. (legorrec@femto-st.fr).}
\address[Third]{Laboratoire d'Automatique et G\'{e}nie des Proc\'{e}d\'{e}s CNRS UMR 5007, Universit\'{e} de Lyon, Universit\'{e} Lyon 1, F-69622 Villeurbanne, France (maschke@lagep.univ-lyon1.fr)}

\begin{abstract}
Boundary controlled irreversible port-Hamiltonian systems (BC-IPHS) on 1-dimensional spatial domains are defined by extending the formulation of reversible BC-PHS to irreversible thermodynamic systems controlled at the boundaries of their spatial domains. The structure of BC-IPHS has clear physical interpretation, characterizing the coupling between energy storing and energy dissipating elements. By extending the definition of boundary port variables of BC-PHS to deal with the dissipative terms, a set of boundary port variables are defined such that BC-IPHS are passive with respect to a given set of conjugated inputs and outputs. As for finite dimensional IPHS, the first and second principle are satisfied as a structural property. Several examples are given to illustrate the proposed approach.
\end{abstract}


\end{frontmatter}



\section{Introduction}

The control of processes in Chemical Engineering is a highly difficult problem due to the nonlinearities induced as well by their thermodynamic properties as their flux relations. One very fruitful approach for the synthesis of nonlinear controllers is to use the properties of the dynamical models arising from first principle modeling such as symmetries, invariants and more generally balance equations of particular thermodynamic potential functions such as the entropy. These balance equations may be used as dissipation inequalities in passivity-based control as introduced in \cite{Willems_72} and is now a well-developed branch of control \citep{Schaft:book,BrogliatoSpringer20}.

In the case of chemical engineering processes, various thermodynamic potentials, such as the entropy or Helmholtz free energy, may be used as storage functions in a dissipation inequality \citep{AlYd:96} and may be used for control design methods based on Lyapunov control functions \citep{CHRISTOFIDES_1998,Christofides_2001} and passivity \citep{AlYd:96,Alonso_Ydstie_01,Alonso_2002,SCHAUM_2018}.

The derivation of these Lyapunov functions and control Lyapunov functions are in most cases, based on the axioms of Equilibrium and Irreversible Thermodynamics and the structure of the dynamical models for these systems. A variety of such "thermodynamic" dynamical models have been suggested which are generalization of gradient control systems \citep{CortesCrouch05}, Lagrangian control systems \citep{OrtegaBook98}, Hamiltonian control systems \citep{Brockett77,Schaft86,Nijmeijer_book_1990}, \citep[chap. 7]{MarsdenBook92} or Port Hamiltonian systems \citep{maschkeNOLCOS92,schaftAEU95,Geoplex09,Jeltsema_FTSC_14} in the sense that they should account both for the conservation of the total energy and for the irreversible entropy production.

A first class of these thermodynamic control systems is defined by pseudo-gradient systems \citep{FaDo:2010,Favache2011}, meaning that they are redefined with respect to a pseudo-metric, in a very similar way as suggested for electrical circuits in \cite{Brayton_Moser_I_1964,smale72}. A second class of systems is defined as metriplectic systems (sum of Hamiltonian and gradient systems) with one or two generating functions \citep{GENERIC_I_PhysRevE, GENERIC_II_PhysRevE, Muschik00, Ramirez_JPCONT_2009, Hoang2011, Hoang_JPC_2012}. A third class of systems is defined as nonlinearly constrained Lagrangian systems \citep{Gay_Balmaz_Entropy_2018_DiscrOpenTherm}. A fourth class of systems is defined as implicit Hamiltonian control systems, in the sense that they are defined on a submanifold of some embedding space (the Thermodynamic Phase space or its symplectic extension), by control Hamiltonian systems defined on contact manifolds \citep{MrNuScSa:91,EberardRMP07, Favache_CES10,Ramirez_CES_2013,Ramirez_SCL_2013,Ramirez_TAC_2017,Merker_ContMechTherm_13} or their symplectization \citep{Schaft_Maschke_Entropy_2018}.

In this paper we shall be interested in another class, namely the Irreversible port-Hamiltonian System (IPHS) which were suggested in \citep{Ramirez_CES_2013} as an extension of port-Hamiltonian systems using a quasi-Poisson bracket and still embedding the irreversible entropy creation. For processes described by lumped models it has been shown that the formalism encompasses a large and general class of irreversible thermodynamic systems, such as heat-exchangers, chemical reactions, chemical reaction networks and coupled mechanic-thermodynamic systems \citep{Ramirez_EJC_2013}. Moreover, using the definition of the availability function \citep{Keenan_1951,Alonso_Ydstie_01} the IPHS structure has recently been employed to exploit the thermodynamic properties of irreversible processes to derive non-linear passivity based controllers \citep{Ramirez_Automatica_2016} .

This paper extends the IPHS formulation to systems defined on 1-dimensional spatial domains building on the framework of boundary controlled PHS \citep{SchaftMaschke2002} using the formulation on 1D spaces \citep{gorrec:2005,Jacob_Zwart_2012} with the aim of using the passivity-based and geometric control design methods, such as methods based on invariance \citep{GODASI_2002}, Casimir functions \citep{Macchelli_TAC_2017,Macchelli_TAC_2020} or (non)-linear dynamic boundary control \citep{Ramirez_TAC_2014,Ramirez_Automatica_2017}. Note that a first approach in this line was given in \citep{Ramirez_2016_TFMST} for a diffusion process.

The paper is organized as follows. In Section \ref{section_p-system} a 1-D isentropic fluid, with and without dissipation, is presented as motivating example to contextualize the proposed irreversible model. Section \ref{section_BC_IPHS} presents the main contribution of the paper, namely the definition of boundary controlled IPHS on a 1D spatial domain. Section \ref{section_energy_entropy} gives two important lemmas regarding the passivity of the system, namely the energy conservation and the irreversible entropy production. In Section \ref{section_diffusion_reaction} the IPHS model of a general diffusion-reaction process is presented, and finally in Section \ref{section_conclusion} some conclusions and comments on future work are given. An Appendix with the basic definitions of boundary controlled PHS is also included. 

\section{Motivating example: 1-D  compressible fluid}\label{section_fluids}\label{section_p-system}

\subsection{The isentropic fluid: the reversible case}
Let us first consider the dynamic behavior of an 1-D isentropic fluid in Lagrangian coordinates, also known as \emph{p-system}, and recall its Port Hamiltonian formulation \citep{Maschke_2005_Springer_Chapter}. The 1-D spatial domain is the interval $[a,b] \ni z,\,a,\,b\in\mathbb{R},\,a<b$. Using as state variables the specific volume $\phi(t,z)$ and the velocity $\upsilon(t,z)$  of the fluid, the dynamical model of the fluid is given by the mass balance equation (expressed in terms of the specific volume) and the momentum balance equation (expressed in terms of the velocity seen as "momentum density") 
\begin{align}\label{isentfluid1}
\frac{\partial \phi}{\partial t} (t,z)& = \frac{\partial \upsilon}{\partial z}(t,z) \\ \label{isentfluid2}
\frac{\partial \upsilon}{\partial t}(t,z)& = -\frac{\partial p}{\partial z}(t,z) 
\end{align}
where $p(\phi)$ is the pressure of the fluid. The Hamiltonian formulation is obtained by considering the total energy of the system which consists in the sum of the kinetic and the internal energy, denoting the internal energy density by  $u(\phi)$
\[
H\left(\upsilon,\, \phi\right)=\int_a^b \left( \frac{1}{2} \upsilon^2+ u(\phi)\right) d z
\]
The variational derivative of the total energy yields  $\frac{\delta H}{\delta \upsilon}=\upsilon$ and $\frac{\delta H}{\delta \phi}=\frac{\partial u}{\partial \phi}=-p$ and the system (\ref{isentfluid1})-(\ref{isentfluid2}) may be written as the \emph{Hamiltonian system}
\begin{equation}\label{p_system_HamiltonianSys}
\begin{bmatrix}
	\frac{\partial \phi}{\partial t}\\
	\frac{\partial \upsilon}{\partial t}
\end{bmatrix}=
P_1\frac{\partial}{\partial z} 
\left(\begin{bmatrix}
	\frac{\delta H}{\delta \phi} \\
	\frac{\delta H}{\delta \upsilon} 
\end{bmatrix}\right)
\end{equation}
where $P_1=\begin{bmatrix}
0 & 1 \\
1 & 0
\end{bmatrix}$ 
and $P_1\frac{\partial}{\partial z}$ is a \emph{Hamiltonian operator} \citep{olver93book,SchaftMaschke2002}. 
Considering an open system, when there is mass and energy flow through the boundary (at the points $a$ and $b$), the Hamiltonian system (\ref{p_system_HamiltonianSys}) is completed with a pair of \emph{boundary port input and output} $\left(v,y\right)$ as follows
\begin{equation*}
    \begin{bmatrix}
    	v \\ y    \end{bmatrix} =
     \begin{bmatrix}
     	W_{B} \\  W_{C}
    \end{bmatrix}
    \begin{bmatrix}
    	 \frac{\delta H}{\delta \phi}(b) \\
     	\frac{\delta H}{\delta v}(b) \\
     	\frac{\delta H}{\delta \phi}(a) \\
     	\frac{\delta H}{\delta v}(a) \\
    \end{bmatrix} = 
    \begin{bmatrix}
    1 & 0 & 0 & 0 \\
    	0 & 0 & -1 & 0 \\
	 0 & 1 & 0 & 0 \\
   	0 & 0 & 0 & 1
    \end{bmatrix}
    \begin{bmatrix}
    	 -p(t,b) \\
	 \upsilon(t,b) \\
    	-p(t,a)\\ \upsilon(t,a)
    \end{bmatrix}
\end{equation*}
from where $v(t) =\begin{bmatrix}
-p(t,b)\\
p(t,a)
\end{bmatrix}$ and $y(t)=\begin{bmatrix}
\upsilon(t,b) \\
\upsilon(t,a) 
\end{bmatrix}$.
It is direct to verify that this choice of inputs and outputs satisfies (\ref{W_B-W_C-BC}), hence the change of energy of the system is given by $\dot{H}(t)=y^\top(t) v(t)$.

\subsection{The isentropic fluid: the irreversible case}
Consider that there is dissipation in the system given by viscous damping. The balance equations are then given by 
\begin{align} 
\frac{\partial \phi}{\partial t} (t,z)& = \frac{\partial \upsilon}{\partial z}(t,z) \label{nonisentfluid1} \\
\frac{\partial \upsilon}{\partial t}(t,z)& = -\frac{\partial p}{\partial z}(t,z) -\frac{\partial \tau}{\partial z}(t,z) \label{nonisentfluid2}
\end{align}
where $\tau$ is the viscous tensor defined as $\tau = -\hat{\mu}\frac{\partial \upsilon}{\partial z}$, with $\hat{\mu}$ the viscous damping coefficient.
The system contains dissipation or rather a irreversible phenomenon induced by the viscosity of the fluid. Therefore we account for the thermal domain and consider Gibbs' equation $du=-pd\phi+Tds$ where $s$ denotes the entropy density and $T$ the temperature. The total energy of the system is still the sum of the kinetic and the internal energy
\[
H\left(\upsilon,\, \phi, \, s\right)=\int_a^b \left( \frac{1}{2} \upsilon^2+ u \left( \phi, s \right) \right) d z
\]
The mass balance (\ref{nonisentfluid1}) and momentum balance equations (\ref{nonisentfluid2}) are then augmented with the entropy balance equation
\begin{equation*}
\frac{\partial s}{\partial t} (t,z) =\frac{\hat{\mu}}{T} \left(\frac{\partial \upsilon}{\partial z}\right)^2(t,z) 
\end{equation*}
and the system of balance equations may be written as the quasi-Hamiltonian system
\begin{equation*}
\begin{bmatrix}
\frac{\partial \phi}{\partial t}\\
\frac{\partial \upsilon}{\partial t}\\
\frac{\partial s}{\partial t}
\end{bmatrix} = 
\begin{bmatrix}
0 & \frac{\partial \left(\cdot \right)}{\partial z}  & 0 \\
\frac{\partial \left(\cdot\right)}{\partial z}  & 0 & \frac{\partial}{\partial z}\left(\frac{\hat{\mu}}{T} \left(\frac{\partial \upsilon}{\partial z}\right)\left(\cdot \right)\right) \\
0 & \frac{\hat{\mu}}{T} \left(\frac{\partial \upsilon}{\partial z}\right)\frac{\partial \left(\cdot \right) }{\partial z} & 0 \\
\end{bmatrix}
\left(\begin{bmatrix}
\frac{\delta H}{\delta \phi} \\
\frac{\delta H}{\delta \upsilon}\\
\frac{\delta H}{\delta s} 
\end{bmatrix}\right)
\end{equation*}
As the operator depends on the co-energy variable $T$, it is only a quasi-Hamiltonian operator.
 
In the following section this latter formulation will be used to define boundary controlled irreversible port-Hamiltonian systems, extending the framework originally proposed in \cite{Ramirez_CES_2013} for irreversible thermodynamic systems on finite dimensional spaces to systems defined on infinite dimensional spaces.   

\section{Boundary controlled IPHS}\label{section_BC_IPHS}

In this section, we introduce the Irreversible Boundary Port Hamiltonian System (IPHS) defined on a 1D spatial domain $z \in [a,b],\,a,\,b\in\mathbb{R},\,a<b$. The state variables of the system are the $n+1$ \emph{extensive variables}\footnote{A variable is qualified as extensive when it characterizes the thermodynamic state of the system and its total value is given by the sum of its constituting parts.}. The following partition of the state vector shall be considered: the first $n$ variables by $x=[q_1,\ldots,q_n]^\top \in\mathbb{R}^n$ and the entropy density by $s \in \mathbb{R}$. The thermodynamic properties of the system are expressed by Gibbs' equation \citep{Callen85}, which we give here in its local form with pairs of specific energy-conjugated variables \citep[Chapter 3]{Geoplex09} 
\begin{equation*}
dh=Tds + p_i\sum_{i=1}^ndq_i
\end{equation*} 
where $T$ is the temperature, conjugated to the entropy density, and the variables $p_i$ denote the \emph{intensive variables}, which are conjugated to the $q_i$ variables. Gibbs' equation is here understood in a general context in order to account for coupled thermo-electro/magnetic/mechanical systems. Gibbs' equation is equivalent to the existence of an energy functional
\begin{equation*}
H(x,s)=\int_a^b h\left(x(z),s(z)\right) dz
\end{equation*} 
where $h(x,s)$ is the energy density function. The total entropy functional is denoted by
\begin{equation*}
S(t)=\int_a^b s(z,t) dz
\end{equation*}
The following pseudo (locally defined) brackets will be used to define the thermodynamic driving forces of the system
\begin{equation}\label{pseudo-bracket}
\begin{split}
\left\{\Gamma|{\mathcal G}|\Omega \right\} & =
\begin{bmatrix}
\frac{\delta \Gamma}{\delta x} \\ 
\frac{\delta \Gamma}{\delta s}
\end{bmatrix}
\begin{bmatrix}
0 & {\mathcal G}\\
-{\mathcal G}^* & 0 \end{bmatrix}
\begin{bmatrix}
\frac{\delta \Omega}{\delta x} \\ 
\frac{\delta \Omega}{\delta s}
\end{bmatrix}, \\ 
\left\{\Gamma|\Omega\right\} & = \frac{\delta \Gamma}{\delta s}^\top \left( \frac{\partial}{\partial z}\frac{\delta \Omega}{\delta s}\right)
\end{split}
\end{equation} 
for some smooth functions $\Gamma$, $\Omega$ and $\mathcal G$.

We shall first define a system of balance equations in terms of an Irreversible (quasi-)Hamiltonian system.
\begin{definition}\label{definition_IPHS}
An infinite dimensional IPHS undergoing $m$ irreversible processes is defined by the PDE
\begin{multline}\label{IPHS-BCS}
\frac{\partial }{\partial t} 
\begin{bmatrix} 
x(t,z)\\ 
s(t,z)
\end{bmatrix}=
\begin{bmatrix}
P_0 & G_0\mathbf{R_0(x)}\\ 
-\mathbf{R_0(x)^\top }G_0^\top &0
\end{bmatrix}
\begin{bmatrix}
\frac{\delta H}{\delta x}(t,z)\\ 
\frac{\delta H}{\delta s}(t,z) 
\end{bmatrix}+
  \\ \begin{bmatrix}P_1 \frac{\partial (.)}{\partial z}  & \frac{\partial \left(G_1 \mathbf{R_1(x)} .\right) }{\partial z}\\ \mathbf{R_1(x)}^\top G_1^\top  \frac{\partial\left(.\right) }{\partial z}&g_s \mathbf{r_s(x)} \frac{\partial \left(.\right)}{\partial z}+\frac{\partial \left(g_s \mathbf{r_s(x)}.\right)}{\partial z}\end{bmatrix}\begin{bmatrix}\frac{\delta H}{\delta x}(t,z)\\ \frac{\delta H}{\delta s}(t,z) \end{bmatrix}
\end{multline}
where $P_0=-{P}^\top_0 \in \mathbb{R}^{n\times n}$, $P_1={P}^\top_1 \in \mathbb{R}^{n\times n}$, $g_s  \in \mathbb{R}$, $G_0 \in \mathbb{R}^{n\times m}$, $G_1 \in \mathbb{R}^{n\times m}$ with $m$ the number of states involved in the entropy production. $\mathbf{R_0}\in \mathbb{R}^{m\times 1}$, $\mathbf{R_1}\in \mathbb{R}^{m\times 1}$ and $r_{s}\in \mathbb{R}$ stand for the vectors of modulated driving forces with  
\begin{equation*}
R_{0,i}=\gamma_{0,i}\left(x,z,\tfrac{\delta H}{\delta x}\right) \left\{S|G_0(:,i)|H\right\}
\end{equation*}
\begin{equation*}
R_{1,i}=\gamma_{1,i}\left(x,z,\tfrac{\delta H}{\delta x}\right) \left\{S|G_1(:,i)\tfrac{\partial}{\partial z}|H\right\}
\end{equation*}
and  
\begin{equation*}
r_{s}=\gamma_s\left(x,z,\tfrac{\delta H}{\delta x}\right) \left\{S|H\right\}
\end{equation*}
with $\gamma_{k,i}\left(x,z,\tfrac{\delta H}{\delta x}\right), \gamma_{s}\left(x,z,\tfrac{\delta H}{\delta x}\right):\mathbb{R}^{n}\rightarrow\mathbb{R}$, $\gamma_{k,i}, \gamma_{s}\geq0$, non-linear positive functions.
\end{definition}
\begin{remark}
Definition \ref{definition_IPHS} is an extension of the definition of IPHS for finite dimensional systems presented in \cite{Ramirez_CES_2013,Ramirez_EJC_2013} for thermodynamic systems defined on 1D spatial domains. In fact it is not difficult to verify that if finite dimensional system is considered only the zero order operator matrix is present in (\ref{IPHS-BCS}) and Definition \ref{definition_IPHS} reduces to 
\begin{equation*}
\frac{d}{d t} 
\begin{bmatrix} 
x(t,z)\\ 
s(t,z)
\end{bmatrix}=
\begin{bmatrix}
P_0 & G_0\mathbf{R_0(x)}\\ 
-\mathbf{R_0(x)^\top }G_0^\top &0
\end{bmatrix}
\begin{bmatrix}
\frac{\delta H}{\delta x}(t,z)\\ 
\frac{\delta H}{\delta s}(t,z) 
\end{bmatrix}
\end{equation*}
which is equivalent to the definition in \cite{Ramirez_CES_2013,Ramirez_EJC_2013} for the case $m=1$ or \cite{Ramirez_MTNS_2014,Ramirez_Automatica_2016} for $m>1$.
\end{remark}
\begin{remark}
The pseudo-brackets (\ref{pseudo-bracket}) define the thermodynamic driving forces of the process. Notice that when considering a finite dimensional system (\ref{pseudo-bracket}) is equal to the pseudo-bracket defined in \cite{Ramirez_CES_2013,Ramirez_EJC_2013} for lumped IPHS.
\end{remark}
In the following definition the above system is completed with boundary port variables. 
\begin{definition}\label{definition_BC_IPHS}
A boundary controlled IPHS (BC-IPHS) is an infinite dimensional IPHS according to Definition \ref{definition_IPHS} equipped with boundary inputs and outputs, defined as the linear combinations of the boundary port variables, respectively
\begin{align*}\label{eq:IPHS_input}
v(t)  & = W_{B}
\begin{bmatrix}
e(t,b) \\
e(t,a)
\end{bmatrix},&  &y(t)  = W_{C}
\begin{bmatrix}
e(t,b) \\
e(t,a)
\end{bmatrix} 
\end{align*}
where the boundary port variables are
\begin{equation}\label{boundary_variables_IPHS}
e(t,z)  = 
\begin{bmatrix}
\frac{\delta H}{\delta x}(t,z) \\
 \mathbf{R(x)}\frac{\delta H}{\delta s}(t,z) \\
\end{bmatrix},
\end{equation}
with $\mathbf{R(x)}=\begin{bmatrix}1 & \mathbf{R_1(x)} & \mathbf{r_s(x)} \end{bmatrix}^\top$ and
\begin{align*}
W_{B} &= 
\begin{bmatrix}
\frac{1}{\sqrt{2}}\left(\Xi_2 + \Xi_1P_{ep} \right)M_p & \frac{1}{\sqrt{2}} \left(\Xi_2 - \Xi_1 P_{ep} \right)M_p
\end{bmatrix},
\\
W_{C}& = 
\begin{bmatrix} 
\frac{1}{\sqrt{2}}\left(\Xi_1+\Xi_2 P_{ep}\right)M_p  & \frac{1}{\sqrt{2}}\left(\Xi_1-\Xi_2 P_{ep}\right)M_p 
\end{bmatrix},
\end{align*}   
where $M_p=\left( M^\top M\right)^{-1}M^\top$, $P_{ep}=M^\top P_e M$ and $M$ is spanning the columns of $P_e$, defined by\footnote{$0$ has to be understood as the zero matrix of proper dimensions.}
\begin{equation}\label{def:Pe}
P_e=\begin{bmatrix}
 P_1 & 0 &G_1 & 0\\
0 & 0 &0 & g_s\\
G_1^\top & 0 &0 & 0\\
0 & g_s  & 0  & 0\\
\end{bmatrix}
\end{equation}
and where $\Xi_1$ and $\Xi_2$ satisfy $\Xi_2^\top\Xi_1+\Xi_1^\top\Xi_2=0$ and $\Xi_2^\top\Xi_2+\Xi_1^\top \Xi_1=I$.\rule{0.5em}{0.5em}
\end{definition}

\begin{remark}
Definition \ref{definition_BC_IPHS} proposes an extension of the boundary port variables defined in \cite{gorrec:2005} for (reversible) BC-PHS to deal with non-reversible systems. Indeed, the vector of boundary port-variables (\ref{boundary_variables_IPHS}) is composed by the vector of intensive variables and the vector of modulated driving forces. On the other hand the matrix (\ref{def:Pe}), which is defined by the reversible structure matrix $P_1$ and the irreversible structure matrices $G_1$ and $g_s$, define the admissible parametrization to obtain the boundary inputs and outputs. Notice that when no irreversible phenomena is present, and thus the entropy coordinate is not considered, Definition \ref{definition_IPHS} and \ref{definition_BC_IPHS} define a BC-PHS \citep{gorrec:2005} (see also the \ref{Appendix}).
\end{remark}

\begin{example}
Recalling the 1D fluid model in Section \ref{section_fluids}, its BC-IPHS formulation is given by $P_0=0, G_0=0, g_s=0$, $P_{1}= 
\begin{bmatrix}
0 & 1  \\
1 & 0 \\
\end{bmatrix}$ and $G_{1} =
\begin{bmatrix}
 0 \\
 1 \\
\end{bmatrix}$
with  $x=\begin{bmatrix}\phi \\ \upsilon\end{bmatrix}$ and $R_{11} = \gamma_{1}\{S | G_1(:,1)\frac{\partial }{\partial z}| H \}$ with $\gamma_{1}=\tfrac{\hat{\mu}}{T}>0$. In this case $n=2$, $m=1$ and 
\[
P_{e}=\begin{bmatrix}
0 &  1 &0& 0 &0\\
 1& 0 &0& 1 &0 \\
0& 0 &0& 0 &0 \\
0& 1 &0& 0 &0 \\
0& 0 &0& 0 &0 
\end{bmatrix}
\]
which gives $M=\begin{bmatrix} \frac{1}{2}&0&0&\frac{1}{2}&0\\0&1&0&0&0 \end{bmatrix}^\top$, $M_P=\begin{bmatrix} 0&1&0&0&0\\ 1&0&0&1&0\\ \end{bmatrix}$ and $P_{ep}=\begin{bmatrix}0 & 1\\ 1&0\end{bmatrix}$.
Choosing the parametrization
\[\Xi_1=\frac{1}{\sqrt{2}}\begin{bmatrix}
1 &  0 \\
1& 0 
\end{bmatrix}, \quad \Xi_2=\frac{1}{\sqrt{2}}\begin{bmatrix}
0 &  1 \\
0& -1 
\end{bmatrix}
\]
define the following boundary inputs and outputs
\begin{align*}
v(t)& =\begin{bmatrix}
-p(t,b)+\frac{\hat{\mu}}{T}\frac{\partial{\upsilon}}{\partial z}(t,b)\\
p(t,a)-\frac{\hat{\mu}}{T}\frac{\partial{\upsilon}}{\partial z}(t,a)\\
\end{bmatrix}, & 
y(t)&=\begin{bmatrix}
\upsilon(t,b) \\
\upsilon(t,a)\\ 
\end{bmatrix}.
\end{align*}
As for the reversible case, the boundary inputs and outputs correspond, respectively, to the pressure and the velocities evaluated at points $a$ and $b$. If there is no dissipation in the system, then the boundary inputs and outputs are exactly the same as for the reversible case. It is direct to verify that the internal energy balance is given by $\dot{H}(t)=y(t)^\top v(t)$. Notice that the chosen parametrization is not unique and that a different choice will lead to different boundary inputs/outputs.
\end{example}


\begin{example}
Consider the heat conduction with heat diffusion over a 1D spatial domain. The conserved quantity is the density of internal energy and the state reduces to a unique variable. Choose the internal energy density $u = u(s)$ as thermodynamic potential function (and $U(s)=\int_a^b u dz $), in this case Gibbs relation defines the temperature as intensive variable conjugated to the extensive variable, the entropy by $T=\frac{du}{ds}(s)$. This leads to write the following entropy balance equation \citep{Geoplex09}
\begin{equation*}
\frac{\partial s}{\partial t} = -\frac{1}{T} \frac{\partial}{\partial z}\left( -\lambda\frac{\partial T}{\partial z} \right)\\
\end{equation*}
where $\lambda$ denotes the heat conduction coefficient and $-\lambda\frac{\partial T}{\partial z}=f_Q$ corresponds to the heat flux. Alternatively the heat conduction can be written in terms of the entropy flux $f_S=\frac{1}{T}f_Q = -\frac{\lambda}{T}\frac{\partial T}{\partial z}$,
\begin{equation}\label{heat_eq_ds}
\frac{\partial s}{\partial t} =\frac{\partial}{\partial z}\left(\frac{\lambda}{T}\frac{\partial T}{\partial z} \right) + \frac{\lambda}{T^2}\left(\frac{\partial T}{\partial z}\right)^2 \\
\end{equation}
from where the entropy production $\sigma_s =\frac{\lambda}{T^2}\left(\frac{\partial T}{\partial z}\right)^2$ is directly identified. Recalling that $\frac{\delta U}{\delta s}=T$, the IPHS formulation of the heat conduction is directly obtained from (\ref{heat_eq_ds}),
\begin{equation*}
\frac{\partial s}{\partial t} = \frac{\lambda}{T^2}\frac{\partial T}{\partial z}\frac{\partial}{\partial z} \left(\frac{\delta U}{\delta s}\right) + \frac{\partial}{\partial z} \left(\frac{\lambda}{T^2}\frac{\partial T}{\partial z}\left(\frac{\delta U}{\delta s}\right)\right)
\end{equation*}
which is equivalent to \eqref{IPHS-BCS} where $P_0=0$, $P_1=0$, $G_0=0$, $G_1=0$, $g_s=1$ and $r_s=\gamma_s \{S | U \}$ with $\gamma_s=\frac{\lambda}{T^2}$ and $\{S | U \}=\frac{\partial T}{\partial z}$. In this case $P_e=\frac{1}{2}\begin{bmatrix}0&1\\ 1& 0\end{bmatrix}$, $n=1$ and $m=1$.  Choosing $\Xi_1=\frac{1}{\sqrt{2}}\begin{bmatrix}
1 &  0 \\
1& 0 
\end{bmatrix}$, $\Xi_2=\frac{1}{\sqrt{2}}\begin{bmatrix}
0 &  1 \\
0& -1 
\end{bmatrix}$ the boundary inputs and outputs of the system are
\begin{align*}
v(t)&=
\begin{bmatrix}
\left( \frac{\lambda_s}{T} \frac{\partial T}{\partial z}\right)(t,b) \\ 
-\left(\frac{\lambda_s}{T} \frac{\partial T}{\partial z}\right)(t,a)
\end{bmatrix}, & 
y(t)&=
\begin{bmatrix}
T(t,b)\\ 
T(t,a)
\end{bmatrix},
\end{align*}
respectively the entropy flux and the temperature at each boundary.
\end{example}

\section{Energy and entropy balance equations}\label{section_energy_entropy}

BC-IPHS encode the first and second principle of Thermodynamics, i.e., the conservation of the total energy and the irreversible production of entropy as stated in the following lemmas.

\begin{lemma}\label{lemma_Conservation_of_internalenergy} (Conservation of energy)
The total energy balance is 
\begin{equation*}
\dot{H} = y(t)^\top v(t)
\end{equation*}
which leads, when the input is set to zero, to $\dot{H} = 0$ in accordance with the first principle of Thermodynamics.
\end{lemma}

\begin{proof}
The variation of the total energy with respect to time is
\begin{align*}
\dot{H} & = \int_a^b \frac{\partial h}{\partial t} dz = \int_a^b \left(\frac{\delta H}{\delta x}^\top\frac{\partial x}{\partial t}+\frac{\delta H}{\delta s}^\top\frac{\partial s}{\partial t}\right) dz \\
&=\int_a^b  \begin{bmatrix}\frac{\delta H}{\delta x}(t,z)^\top & \frac{\delta H}{\delta s}(t,z) \end{bmatrix} {\mathcal{J}}_e\begin{bmatrix}\frac{\delta H}{\delta x}(t,z)\\ \frac{\delta H}{\delta s}(t,z) \end{bmatrix} dz
\end{align*}
with
\[
{\mathcal{J}}_e=\begin{bmatrix}P_1 \frac{\partial (.)}{\partial z}  & \frac{\partial \left(G_1 \mathbf{R_1(x)} .\right) }{\partial z}\\ \mathbf{R_1(x)}^TG_1^T \frac{\partial\left(.\right) }{\partial z}&g_s \mathbf{r_s(x)} \frac{\partial \left(.\right)}{\partial z}+\frac{\partial \left(g_s\mathbf{ r_s(x)}.\right)}{\partial z}\end{bmatrix}
\]
where we have used the skew symmetry of the matrix of zero order operators
\[
\begin{bmatrix}
P_0 & G_0\mathbf{R_0(x)}\\ 
-\mathbf{R_0(x)}^TG_0^T&0 
\end{bmatrix}.\]
Noticing that
\[
\int_a^b \frac{\delta H}{\delta x}^\top P_1 \frac{\partial }{\partial z}\left( \frac{\delta H}{\delta x}\right) dz = \frac{1}{2} \left[ \frac{\delta H}{\delta x}^\top P_1 \frac{\delta H}{\delta x}\right]_a^b
\]
that
\begin{multline*}
\int_a^b \left( 
\frac{\delta H}{\delta s} \mathbf{R_1(x)}^TG_1^T \frac{\partial}{\partial z} \left(\frac{\delta H}{\delta x}\right)+\frac{\delta H}{\delta x}^\top  \frac{\partial}{\partial z}\left(G_1\mathbf{R_1(x)} \frac{\delta H}{\delta s}
\right) \right)dz \\
= \left[ \frac{\delta H}{\delta s} \mathbf{R_1(x)}^\top G_1^\top \frac{\delta H}{\delta x}  \right]_a^b
\end{multline*}
and that
\begin{multline*}
\int_a^b \left( 
 \frac{\delta H}{\delta s} g_s \mathbf{r_s(x)} \frac{\partial }{\partial z}\left( \frac{\delta H}{\delta s}\right)+ \frac{\delta H}{\delta s} \frac{\partial }{\partial z}\left(g_s \mathbf{ r_s(x)} \frac{\delta H}{\delta s}\right)
\right)dz \\
= \left[ \frac{\delta H}{\delta s}g_s \mathbf{ r_s(x)} \left(\frac{\delta H}{\delta s}\right) \right]_a^b
\end{multline*}
we have
\begin{align*}
\dot{H} & = \left[  \begin{bmatrix}\frac{\delta H}{\delta x}\\ \frac{\delta H}{\delta s}\\ \mathbf{R_1(x)}G_1^T \frac{\delta H}{\delta s}\\ g_s\mathbf{r_s(x)} \frac{\delta H}{\delta s}\end{bmatrix}^\top P_e   \begin{bmatrix}\frac{\delta H}{\delta x} \\ \frac{\delta H}{\delta s} \\ \mathbf{R_1(x)}G_1 \frac{\delta H}{\delta s}\\ g_s \mathbf{r_s(x)} \frac{\delta H}{\delta s}\end{bmatrix}\right]_a^b
\end{align*}
with $P_e$ defined in (\ref{def:Pe}). Using the parametrization proposed in \citep{gorrec:2005,legorrecCCA2006} ($P_e$ is potentially not full rank), it is possible to write
\[
\begin{bmatrix}
u(t) \\
y(t)
\end{bmatrix}=\frac{1}{\sqrt{2}}\begin{bmatrix}
\Xi_1 & \Xi_2 \\
\Xi_2 & \Xi_1
\end{bmatrix}
\begin{bmatrix}
P_{1p}M_p & -P_{1p}M_p \\
M_p&M_p
\end{bmatrix}
\begin{bmatrix}
e(t,b) \\
e(t,a)
\end{bmatrix} 
\]
with $\Xi_i$, $M_P$ and $P_{1p}$ defined in Definition \ref{definition_BC_IPHS}, from where it is obtained that $\dot{H}=y(t)^\top u(t)$. 
\end{proof}

\begin{lemma}\label{lemma_Irreversible_entropy_production} (Irreversible entropy production)
The total entropy balance is given by  
\begin{equation*}
\dot{S} =  \int_a^b \sigma_t dz -y_{S}^\top v_{s}
\end{equation*}
where $y_s$ and $v_s$ are the entropy conjugated input/output and $\sigma_t$ is the total internal entropy production. This leads, when the input is set to zero, to $\dot{S} = \int_a^b \sigma_t dz \geq 0 $ in accordance with the second principle of Thermodynamics.
\end{lemma}

\begin{proof}
Let's consider the total entropy balance 
\begin{align*}
\dot{S} & =\int_a^b \frac{\partial s}{\partial t}dz  \\
& = \int_a^b \left(\mathbf{R_0(x)}^\top G_0^\top  \frac{\delta H}{\delta x}+\mathbf{R_1(x)}^\top G_1^\top  \frac{\partial }{\partial z}\frac{\delta H}{\delta x}+\right.\\
& \left. g_s \mathbf{r_s(x)} \frac{\partial }{\partial z}\frac{\delta H}{\delta s}+\frac{\partial }{\partial z}\left(g_s\mathbf{ r_s(x)}\frac{\delta H}{\delta x}\right)\right) dz
\end{align*}
The first three terms define the internal entropy production related to the operators of order zero and one 
\begin{align*}
\mathbf{R_0(x)}^\top G_0^\top  \frac{\delta H}{\delta x} & = \sum_i^m \left( R_{0,i} (x) G_0(:,i)^\top \frac{\delta H}{\delta x} \right)\\
&= \sum_i^m  \gamma_{0,i} \left\{S|G_0(:,i)|H\right\}^2= \sum_i^m \sigma_{0i}\geq 0,\\
\mathbf{R_1(x)}^\top G_1^\top \frac{\partial }{\partial z}\frac{\delta H}{\delta x}& = \sum_i^m \left( R_{1,i} (x) G_1(:,i)^\top \frac{\partial }{\partial z}\frac{\delta H}{\delta x} \right)\\
& = \sum_i^m  \gamma_i \left\{S|G_1(:,i)\frac{\partial }{\partial z}|H\right\}^2 =  \sum_i^m \sigma_{1i}\geq 0,\\
g_s \mathbf{r_s(x)} \frac{\partial }{\partial z}\frac{\delta H}{\delta s}&= \gamma_{s} \left\{S|H\right\}^2 = \sigma_{s}\geq 0,
\end{align*}
%
%
%
where $\sigma_{0i}$ and $\sigma_{1i}$ are, respectively, the zero and first order internal entropy productions due to the $i$-th irreversible thermodynamic processes, and $\sigma_{s}$ is the internal entropy production due to entropy (heat) flux. Since the total internal entropy production is the sum of the internal entropy production of all irreversible processes $\sigma_t= \sum_i^m \left( \sigma_{0i}+\sigma_{1i}+\sigma_{s}\right)$ we have
\begin{align*}
\dot{S} & = \int_a^b \sigma_t dz + \int_a^b  \frac{\partial }{\partial z}\left(g_s\mathbf{ r_s(x)}\frac{\delta H}{\delta x}\right)dz \\
& =  \int_a^b \sigma_t dz + \left(g_s\mathbf{ r_s(x)}\frac{\delta H}{\delta x}(b,t) - g_s\mathbf{ r_s(x)}\frac{\delta H}{\delta x}(a,t)\right) \\
& =  \int_a^b \sigma_t dz -\left( f_s(b,t) - f_s(a,t)\right).
\end{align*}
from where we have that the supply rate $y_S^\top v_{s}=\left( f_s(b,t) - f_s(a,t)\right)$, representing the entropy flux at the boundaries. Hence, the total entropy variation is equal to the internal entropy production minus what is flowing in/out through the boundaries. 
\end{proof}

\section{Example: the diffusion-reaction process}\label{section_diffusion_reaction}
Diffusion-reaction processes are systems in which the changes in the mole number per unit volume are due to transport of particles, through processes such as diffusion and convection, and due to chemical reactions. This is the case for instance for tubular reactors \citep{Horn_1972, Feinberg_1987,Ar:89,Kondepudi_98,Sa:06}. Consider a diffusion-reaction process without convection involving $n$ species and $j$ chemical reactions, described by the following set of PDEs
\begin{equation}\label{diffusion_general}
\begin{split}
\frac{\partial c_i}{\partial t} & = -\frac{\partial f_{ci}}{\partial z} + \sum_{k=1}^j \bar{\nu}_{ki} r_{k}, \qquad  i=1,\ldots,n \\
\frac{\partial s}{\partial t} & = -\frac{\partial f_s}{\partial z} +\sum_{k=1}^{n} \sigma_{c_k} + \sum_{k=1}^j \sigma_{r_k} + \sigma_{s}.
\end{split}
\end{equation}
where $c_i$ is the  molar concentration per unit volume of the $i$-th specie, $s$ is the entropy density, $f_i=-\frac{L_i}{T}\frac{\partial \mu_i}{\partial z}$ corresponds to the molar flux of the $i$-th specie, $f_s=-\frac{\lambda}{T}\frac{\partial T}{\partial z}$ to the entropy flux, $L_i$ is the diffusion coefficients of the $i$-th specie, $\lambda$ is heat conduction coefficient, $\nu_{ki}$ is the stoichiometric coefficient of the reactant $i$ in the $k$-th reaction and $r_k$ is the reaction rate of the $k$-th reaction. The thermodynamic driving force of the $k$-th chemical reaction is the chemical affinity of the $k$-th reaction $\mathcal{A}_k=-\sum_{i=1}^{n} \bar{\nu}_{ki}\mu_{i}$, where $\mu_i$ is the chemical potential of the $i$-th specie. The internal entropy production of the process is due to irreversible diffusion, heat conduction and to the chemical reactions. In (\ref{diffusion_general}) the internal entropy production terms are respectively, $\sigma_{c_i} =- \frac{1}{T}f_{ci}\frac{\partial \mu_i}{\partial z} = \frac{L_i}{T^2}\left(\frac{\partial \mu_i}{\partial z}\right)^2$ for diffusion of the $i$-th species, $ \sigma_{r_k}= \frac{1}{T}r_{k} \mathcal{A}_{k}= \sum_{i=1}^{n-1} \bar{\nu}_{ki} r_{k} \mu_{i}$ for the $k$-th chemical reaction and $\sigma_s =- \frac{1}{T}f_s\frac{\partial T}{\partial z}= \frac{\lambda}{T^2}\left(\frac{\partial T}{\partial z}\right)^2$ for heat conduction.


\subsection{The IPHS model}

\begin{prop}\label{Proposition_diffusion_reaction_IPHS}
Consider the diffusion-reaction process with $x=[c_{1}, \ldots, c_{n}]^\top \in \mathbb{R}^n$. Then (\ref{diffusion_general}) can be written as the infinite dimensional IPHS (Definition \ref{definition_IPHS})
\begin{multline}\label{IPHS_diffusion_infinite}
\frac{\partial }{\partial t} 
\begin{bmatrix} 
x(t,z)\\ 
s(t,z)
\end{bmatrix}=
\begin{bmatrix}
0 & G_0\mathbf{R_0(x)}\\ 
-\mathbf{R_0(x)^\top }G_0^\top & 0
\end{bmatrix}
\begin{bmatrix}
\frac{\delta H}{\delta x}(t,z)\\ 
\frac{\delta H}{\delta s}(t,z) 
\end{bmatrix} + \\ 
\begin{bmatrix}
0  & \frac{\partial \left(G_1 \mathbf{R_1(x)} \cdot \right) }{\partial z}\\
\mathbf{R_1(x)}^\top G_1^\top  \frac{\partial\left(\cdot \right) }{\partial z} & g_s \mathbf{r_s(x)} \frac{\partial \left( \cdot \right)}{\partial z} + \frac{\partial \left(g_s \mathbf{r_s(x)} \cdot \right)}{\partial z}
\end{bmatrix}
\begin{bmatrix}
\frac{\delta H}{\delta x}(t,z)\\ 
\frac{\delta H}{\delta s}(t,z) 
\end{bmatrix}
\end{multline}
The modulating function of the mass diffusion of the $i$-th species is
\[
R_{1i} = \gamma_{1i}\{S | G_1(:,1) \tfrac{\partial}{\partial z} | U\} =  \frac{1}{T}\left(\frac{L_i}{T}\right)\frac{\partial \mu_i}{\partial z}
\]
with $\gamma_{ci} = \frac{1}{T}\left(\frac{L_i}{T}\right)>0$, $\{S | G_1(:,i) \tfrac{\partial}{\partial z} | U\}=\frac{\partial \mu_i}{\partial z}$ and the matrices of the operators of order one $G_1=I_n$. The modulating function of the heat diffusion is
\[
r_{s} = \gamma_{s}\{S  | U\}= \frac{1}{T}\left(\frac{\lambda}{T}\right)\frac{\partial T}{\partial z}
\] 
with $\gamma_{s} =  \frac{1}{T}\left(\frac{\lambda}{T}\right)>0$, $\{S | U\}=\frac{\partial T}{\partial z}$ and $g_s=1$. The modulating function of the  $j$-th chemical reaction is
\begin{equation*}
R_{0j}=\gamma_{0j}\{s | G_0(:,j) | U \} = \frac{1}{T}\left(\frac{r_j}{\mathcal{A}_j}\right)\mathcal{A}_j
\end{equation*}
with $\gamma_{r_j} = \frac{1}{T}\left(\frac{r_j}{\mathcal{A}_j}\right)>0$, $\{S | G_{0j} | U \}=\mathcal{A}_j$ and the matrix of the operator of order zero $G_0=I_n$. 

Furthermore, consider the following set of boundary inputs and outputs, respectively,
\begin{align}\label{IPHS_inputs_outputs}
v&=
\begin{bmatrix}
\mathbf{f}(t,a) \\
\mathbf{f}(t,b)
\end{bmatrix}, &
y&=
\begin{bmatrix}
-\frac{\delta H}{\delta x}(t,b)\\ -\frac{\delta H}{\delta s}(t,b)\\
\frac{\delta H}{\delta x}(t,a)\\ \frac{\delta H}{\delta s}(t,a)
\end{bmatrix},
\end{align}
with $\mathbf{f}=[f_{c1}, \ldots, f_{cn}, f_s]^\top$ the vector of fluxes, then (\ref{IPHS_diffusion_infinite}) with (\ref{IPHS_inputs_outputs}) is a BC-IPHS (Definition \ref{definition_BC_IPHS}). 
\end{prop}
\begin{proof}
It is straightforward to verify that the pseudo-brackets related to mass and heat diffusion correspond to the respective thermodynamic driving forces. The non-linear functions $\gamma_{ci}$ and $\gamma_s$ are positive since the mass diffusion coefficients $L_i$, the thermal conductivity coefficient $\lambda$ and the temperature are always positive. 
For the thermodynamic driving force of the chemical reaction, consider in a first instance only the $j$-th reaction term. The pseudo-bracket for the $j$-th chemical reaction is
\[
\{S | G_{0j} | U \}= 
\begin{bmatrix}
0 \\
\vdots \\
0 \\
1
\end{bmatrix}^\top
\begin{bmatrix}
0 & \ldots & 0 & \bar{\nu}_{j1} \\
0 & \ldots & 0 & \vdots \\
0 & \ldots & 0 & \bar{\nu}_{j{n}} \\
-\bar{\nu}_{j1} & \ldots & -\bar{\nu}_{j{n}} & 0
\end{bmatrix}
\begin{bmatrix}
\mu_1 \\
\vdots \\
\mu_{n-1} \\
T
\end{bmatrix} = 
\mathcal{A}_j.
\]
On the other hand from De Donder's fundamental equation \citep{Prigogine:52} it has been shown in \citep{Ramirez_CES_2013} that $\gamma_{r_j} = \frac{1}{T}\left(\frac{r_j}{\mathcal{A}_j}\right)>0$. The same applies for the other reactions. Regarding the boundary inputs and outputs, we have that
\begin{equation*}
P_e=\begin{bmatrix}
0 & 0 &I_n & 0\\
0 & 0 &0 & 1\\
I_n & 0 &0 & 0\\
0 & 1  & 0  & 0\\
\end{bmatrix}
\end{equation*}
since $P_1=0, G_1=I_n$ and $g_s=1$. This implies $M_p=P_{ep}=M=P_e$, hence choosing the parametrization
\begin{align*}
\Xi_1=
\frac{1}{\sqrt{2}}\begin{bmatrix}
0 & 0 & -I_n & 0 \\
0 & 0 & 0 & -1  \\
0 & 0 & I_n & 0 \\
0 & 0 & 0 & 1 \\
\end{bmatrix}, &&
\Xi_2=
\frac{1}{\sqrt{2}}\begin{bmatrix}
I_n & 0 & 0 & 0 \\
0 & 1 & 0 & 0 \\
I_n & 0 & 0 & 0 \\
0 & 1 & 0 & 0 \\
\end{bmatrix} 
\end{align*}
we obtain the boundary inputs/outputs (\ref{IPHS_inputs_outputs}).
\end{proof}

\subsection{A simple application case}

Consider a simple diffusion-reaction process on $z \in [a,b]$ involving only two species and obeying the following reaction scheme
\begin{equation}\label{example_A-B}
A \stackrel{r}{\longrightarrow} B 
\end{equation}
This simple reaction involves four irreversible thermodynamic processes, related to the mass diffusion of species $A$ and $B$, the heat diffusion and the chemical reaction. The thermodynamic parameters of the diffusion-reaction process are the mass diffusion coefficients $L_A$ and $L_B$, the thermal conductivity coefficient $\lambda$ and the stoichiometric coefficients $\bar{\nu}_A=-1$ and $\bar{\nu}_B=1$. The state vector is in this case $x=[c_A, c_B]^\top$ and $s$, and according to Proposition \ref{Proposition_diffusion_reaction_IPHS} the BC-IPHS formulation of (\ref{example_A-B}) is
\begin{multline*}
\frac{\partial}{\partial t}
\begin{bmatrix}
c_A \\
c_B \\
s
\end{bmatrix} = 
\begin{bmatrix}
0 & 0 & \frac{r}{T^2} \\
0 & 0 & -\frac{r}{T^2}\\
-\frac{r}{T^2}& \frac{r}{T^2} &  0
\end{bmatrix}
\begin{bmatrix}
\mu_A \\
\mu_B \\
T
\end{bmatrix}+ \\
\begin{bmatrix}
0 & 0 & \frac{\partial}{\partial z}\left( \frac{L_A}{T^2}\frac{\partial \mu_A}{\partial z}\left(.\right)\right) \\
0 & 0 & \frac{\partial}{\partial z}\left( \frac{L_B}{T^2}\frac{\partial \mu_A}{\partial z}\left(.\right)\right) \\
\frac{L_A}{T^2}\frac{\partial}{\partial z}\left(.\right) & \frac{L_B}{T^2}\frac{\partial}{\partial z}\left(.\right) &  \frac{\lambda}{T}\left(\frac{\partial T}{\partial z}\right)\frac{\partial \left(.\right)}{\partial z}+\frac{\partial }{\partial z}\left( \frac{\lambda}{T^2}\frac{\partial T}{\partial z} \left(.\right)\right)
\end{bmatrix}
\begin{bmatrix}
\mu_A \\
\mu_B \\
T
\end{bmatrix}
\end{multline*}

%
%
The conjugated inputs and outputs are, respectively,
\begin{equation*}
v=
\begin{bmatrix}
\frac{L_A}{T}\frac{\partial \mu_A}{\partial z}(t,a) \\
\frac{L_B}{T}\frac{\partial \mu_B}{\partial z}(t,a) \\
\frac{\lambda}{T}\frac{\partial T}{\partial z}(t,a) \\
\frac{L_A}{T}\frac{\partial \mu_A}{\partial z}(t,b) \\
\frac{L_B}{T}\frac{\partial \mu_B}{\partial z}(t,b) \\
\frac{\lambda}{T}\frac{\partial T}{\partial z}(t,b) 
\end{bmatrix}
, \qquad
y=\begin{bmatrix}
-\mu_A(t,b)\\
-\mu_B(t,b)\\
-T(t,b)\\
\mu_A(t,a)\\
\mu_B(t,a)\\
T(t,a)
\end{bmatrix}
\end{equation*}
i.e., the incoming and outgoing flows of matter and entropy evaluated at the boundaries and the intensive variables evaluated at the boundaries. From Lemma \ref{lemma_Conservation_of_internalenergy} the energy balance is given by
\begin{align*}
\dot{U} & =y^\top v \\
& = \left( \frac{L_A}{T}\frac{\partial \mu_A}{\partial z}\mu_A(b) - \frac{L_A}{T}\frac{\partial \mu_A}{\partial z}\mu_A (a) \right) \\
& + \left( \frac{L_B}{T}\frac{\partial \mu_B}{\partial z}\mu_B(b) - \frac{L_B}{T}\frac{\partial \mu_B}{\partial z}\mu_B (a) \right) \\ 
& + \left( \frac{\lambda}{T}\frac{\partial T}{\partial z}T(b) - \frac{\lambda}{T}\frac{\partial T}{\partial z}T(a) \right),
\end{align*}
while by Lemma \ref{lemma_Irreversible_entropy_production} the entropy balance is
\begin{multline*}
\dot{S}=\\
\int_a^b\frac{\lambda}{T^2}\left(\frac{\partial T}{\partial z}\right)^2 + \frac{L_A}{T^2}\left(\frac{\partial \mu_A}{\partial z}\right)^2 + \frac{L_B}{T^2}\left(\frac{\partial \mu_B}{\partial z}\right)^2 + \frac{r}{T^2}\mathcal{A} \; dz \\
+ \left(\frac{\lambda}{T}\frac{\partial T}{\partial z}(b) - \frac{\lambda}{T}\frac{\partial T}{\partial z}(a)\right)
\end{multline*}
We observe that the total internal entropy production is 
\[
\sigma=\int_a^b\frac{\lambda}{T^2}\left(\frac{\partial T}{\partial z}\right)^2 + \frac{L_A}{T^2}\left(\frac{\partial \mu_A}{\partial z}\right)^2 + \frac{L_B}{T^2}\left(\frac{\partial \mu_B}{\partial z}\right)^2 + \frac{r}{T^2}\mathcal{A} \; dz
\] 
Furthermore, $\{S | G_1(:,1)\tfrac{\partial}{\partial z} | U\}=\frac{\partial \mu_A}{\partial z}$, $\{S| G_1(:,2)\tfrac{\partial}{\partial z} |U\}=\frac{\partial \mu_B}{\partial z}$, $\{S | U\}=\frac{\partial T}{\partial z}$ and $\{S | G_0(:,i)|U\}=\mathcal{A}$ correspond, respectively, to the thermodynamic driving forces of mass and heat diffusion and chemical reactions. The diffusion-reaction process is hence given by the composition of the IPHS formulation of the diffusion process and the chemical reaction together with the mass and heat flows.

\section{Conclusion}\label{section_conclusion}
Boundary controlled irreversible port-Hamiltonian systems (BC-IPHS) on 1-dimensional spatial domains have been defined (Definition \ref{definition_IPHS} and \ref{definition_BC_IPHS}) by extending the formulation of reversible BC-PHS to irreversible thermodynamic systems controlled at the boundaries of their spatial domains. The structure of BC-IPHS has clear physical interpretation, characterizing the coupling between energy storing and energy dissipating elements, furthermore, the irreversible nature of the model is precisely expressed by the thermodynamic driving forces. By extending the definition of boundary port variables of BC-PHS to deal with the dissipative terms, a set of boundary port variables have been defined such that BC-IPHS are passive with respect to a given set of conjugated inputs and outputs. It is interesting to notice that when no irreversible phenomena is present, and thus the entropy coordinate is not considered, Definition \ref{definition_IPHS} and \ref{definition_BC_IPHS} define a BC-PHS. As for finite dimensional IPHS, the first and second principle are satisfied (Lemmas \ref{lemma_Conservation_of_internalenergy} and \ref{lemma_Irreversible_entropy_production}) as a structural property. The proposed formulation has been illustrated on the examples of an isentropic fluid, with and withot dissipation, heat conduction and a diffusion-reaction process. Future work will study the extension of passivity based boundary control design methods to BC-IPHS.

\section*{Acknowledgements}
This work has been partially funded by Chilean FONDECYT 1191544 and CONICYT BASAL FB0008 projects, through grants from the European Commision Marie Skodowska-Curie Fellowship, ConFlex ITN Network under reference code 765579 and from the ANR Agency by the EUR EIPHI, and the INFIDHEM project under the reference codes  ANR-16-CE92-0028 and ANR-17-EURE-0002 respectively.

\bibliographystyle{model2-names}
\bibliography{references_HM}

\appendix
\section{Port-Hamiltonian systems on infinite dimensional spaces}
\label{Appendix}

Let us now define a port-Hamiltonian system on an infinite dimensional space. To this end, first we introduce the definition of variational derivative of a functional, see \cite{Nijmeijer_book_1990}.

\begin{definition}\label{definition_variational_derivative}
Consider a functional $H$ defined by
\begin{equation*}
H(x)=\int_a^b h\left(z, x, x^{(1)},\ldots,x^{(n)}\right)dz
\end{equation*}
for any smooth real vector function $x(z)$, $z \in Z=\left(a,b\right)$ where the integrand $u$ is a smooth function of $x$ and its derivatives up to some order $n$. The variational derivative of the functional $H$ is denoted by $\frac{\delta H}{\delta x}$ and is the only function that satisfy for every $\epsilon \in \mathbb{R}$ and smooth real function $\delta x(z)$, $z \in Z$, such that their derivatives satisfy $\delta x^{(i)}(a) = \delta x^{(i)}(b) = 0$, $i = 0, \ldots,n$,
\begin{equation*}
H[x+\epsilon \delta x] = H[x] + \epsilon \int_a^b \frac{\delta H}{\delta x} \delta x dz + O(\epsilon^2)
\end{equation*}
\end{definition}
In the case when $h$ only depends on the vector function $x$ and not its derivatives, then the variational derivative is simply obtained by derivation of the integrand, i.e,
\[
\frac{\delta H}{\delta x} = \frac{\partial h}{\partial x}.
\]

An infinite dimensional PHS on a 1D spatial domain is characterized by the following PDE
\begin{equation}
\label{BCS}
  \frac{\partial x}{\partial t}(t,z) =P_1\frac{\partial}{\partial z}\left(\frac{\delta H}{\delta x}(t,z)\right) + (P_0-G_0)\frac{\delta H}{\delta x}(t,z),
\end{equation}
with $z \in (a,b)$,  $P_1 \in M_n(\mathbb{R})$\footnote{$M_n(\mathbb{R})$ denote the space of real $n\times n$ matrices} a nonsingular symmetric matrix, $P_0=-P_0^\top \in M_n(\mathbb{R})$, $G_0 \in M_n(\mathbb{R})$ with $G_0 \geq 0$ and $x$ taking values in $\mathbb{R}^n$. The controlled (and homogeneous) boundary conditions of (\ref{BCS}) are characterized by a matrix $W_{B}$ of appropriate size such that
\begin{equation*}
    v(t)=W_{B}
    \begin{bmatrix}
     \frac{\delta H}{\delta x}(t,b) \\
     \frac{\delta H}{\delta x}(t,a)
    \end{bmatrix}
  \end{equation*}
Considering the above boundary conditions as the input of the system, we can define an associate boundary output as
  \begin{equation*}
    y(t)=W_{C}
    \begin{bmatrix}
     \frac{\delta H}{\delta x}(t,b) \\
     \frac{\delta H}{\delta x}(t,a)
    \end{bmatrix}.
  \end{equation*}
If $W_{B}$ and $W_C$ satisfy
\begin{equation}\label{W_B-W_C-BC}
\begin{split}
W_{B}\tilde{\Sigma} W_{B}^\top = W_{C}\tilde{\Sigma} W_{C}^\top = 0 \\ 
W_{B}\tilde{\Sigma} W_{C}^\top = W_{C}\tilde{\Sigma} W_{B}^\top = I
\end{split}
\end{equation}
with $\tilde{\Sigma}=\left[\begin{smallmatrix} P_1^{-1} & 0 \\ 0 & -P_1^{-1} \end{smallmatrix}\right]$, it is not difficult to show that under some very general conditions \citep{gorrec:2005, Jacob_Zwart_2012} the change of energy of the system becomes 
\begin{equation*}
    \dot{H}(t) =y^\top(t)v(t)  - \int_{a}^{b} \frac{\delta H}{\delta x}^\top(t,z) G_0 \frac{\delta H}{\delta x}(t,z) dz
\end{equation*}
We can see from this equation that the dissipation in the system is characterized by the matrix $G_0$. Indeed, since the input and output act and sense at the boundary of the spatial domain, in the absence of internal dissipation ($G_0=0$) the system only exchanges energy with the environment through the boundaries. In this case the PHS fullfils
\begin{equation*}
    \dot{H}(t) =y^\top(t)v(t),
\end{equation*}
and the PHS is called conservative. This formulation has proven to be extremely useful to study the existence and uniqueness of solutions for the linear case, and to perform control synthesis for the general class of PHS \citep{gorrec:2005,Jacob_Zwart_2012,Ramirez_TAC_2014,Macchelli_TAC_2017,Ramirez_Automatica_2017}. One interesting feature of PHS is that they are applicable to hyperbolic and parabolic systems, however the PHS formulation of parabolic systems leads necessary to an implicit system.
\end{document}